%% file: main.tex
\documentclass[fleqn]{article}
\usepackage{multicol}
\usepackage{graphicx} 
\usepackage{subfig}
\usepackage{geometry}
\usepackage[utf8]{inputenc}
\usepackage[T1]{fontenc}
\usepackage{setspace}
\usepackage{amsmath}
\usepackage{bbm}
\usepackage{float}
\usepackage{microtype}
\usepackage{amssymb}
\usepackage{amsthm}
\newtheorem{mydef}{Definition}
\usepackage{tikz}
\usepackage{mathtools}
\usepackage{enumerate}
\usepackage{authblk}
\usepackage{hyperref}
\usepackage{xcolor}
\usepackage[
    backend=biber,
    style=phys]{biblatex} %Imports biblatex package
\newcommand{\denop}{\mathcal{D}}
\newcommand{\hs}{\mathcal{H}}

\newcommand{\id}{\mathbbm{1}}

\newcommand{\V}{\mathcal{V}}
\newcommand{\RV}{\mathcal{R_V}}
\DeclareMathOperator{\Tr}{Tr}

\newtheorem{lemma}{Lemma}
\newtheorem{prop}{Proposition}

\title{Local Invariance of Divergence-based Quantum Information Measures}
\author[1]{Christopher Popp}
\author[2]{Tobias C. Sutter}
\author[3]{Beatrix C. Hiesmayr}
\affil[1,2,3]{University of Vienna, Faculty of Physics, Währingerstrasse 17, 1090 Vienna.\vspace{3.5mm}}
\affil[1]{christopher.popp@univie.ac.at}
\affil[2]{tobias.christoph.sutter@univie.ac.at}
\affil[3]{beatrix.hiesmayr@univie.ac.at}
\date{}
\begin{document}

\onehalfspacing
\maketitle

\begin{abstract}
Quantum information quantities, such as mutual information and entropies, are essential for characterizing quantum systems and protocols in quantum information science. In this contribution, we identify types of information measures based on generalized divergences and prove their invariance under local isometric or unitary transformations. Leveraging the reversal channel for local isometries together with the data processing inequality, we establish invariance for information quantities used in both asymptotic and one-shot regimes without relying on the specific functional form of the underlying divergence. 
These invariances can be applied to improve the computation of such information quantities or optimize protocols and their output states whose performance is determined by some invariant measure. Our results improve the capability to characterize and compute many operationally relevant information measures with application across the field of quantum information processing.

\end{abstract}

\input{1_introduction} 
\input{2_local_invariances} 
\input{3_conclusion}
\printbibliography

\section*{Author Contributions} Conceptualization, C.P.; validation, C.P., T.C.S. and  B.C.H.; formal analysis, C.P.; writing---original draft preparation, C.P.; writing---review and editing, C.P., T.C.S. and B.C.H.   All authors have read and agreed to the published version of the manuscript.

\section*{Acknowledgments} This research was funded in whole, or in part, by the Austrian Science Fund (FWF) [10.55776/P36102]. For the purpose of open access, the author has applied a CC BY public copyright license to any Author Accepted Manuscript version arising from this submission.

\newpage

\end{document}

%% file: 1_introduction.tex
\section{Introduction}
A central goal of quantum information theory is the precise quantification of correlations, uncertainties, and distinguishability within quantum systems. Various information-theoretic quantities have been defined to capture the fundamental limits of quantum information processing tasks like communication, computation, or entanglement manipulation.
\\
Traditionally, the fundamental quantity is the von Neumann entropy \cite{neumann_thermodynamik_1927}, from which many quantities relevant in classical information science, like the mutual information or the conditional entropy, can be generalized to the quantum regime (cf. \cite{wilde_quantum_2013}). These measures are relevant due to their operational meaning in quantifying optimal rates of achieving various information processing tasks in the asymptotic regime of infinitely many independent uses of a resource (e.g., quantum state). A notable example is the quantum relative entropy \cite{umegaki_conditional_1962}, which can be related to quantum hypothesis testing \cite{hiai_proper_1991}. Several other information quantities can be expressed in terms of the relative entropy, such as the private information for the rate of secret-key distillation and the coherent information for the rate of entanglement distillation \cite{devetak_distillation_2005}. 
In the so-called one-shot setting (cf. Ref.~\cite{tomamichel_quantum_2016}), only a finite amount of resources are considered. In this approach, the goal is to find optimal rates to transform the limited resources, such as quantum states, into the target states while allowing fixed error bounds. Moreover, the one-shot setting is crucial for security analysis in the context of quantum cryptography without any assumptions on the actions of a malicious third party \cite{renner_security_2006}. As in the asymptotic setting, several information quantities with operational meaning have been identified such as the $\varepsilon$-hypothesis testing mutual information \cite{wang_one-shot_2012} and the smooth max-mutual information \cite{buscemi_quantum_2010} that can be related to secret-key distillation in the one-shot setting.
\\
As quantum technologies advance, the landscape of information quantities and related processing tasks becomes increasingly diverse (see Ref.~\cite{khatri_principles_2024} for a comprehensive overview of information measures and related processing tasks) and a unified framework capable of accommodating both asymptotic and one-shot scenarios is essential. Generalized divergences offer such a unifying language. By definition, each divergence satisfies monotonicity under completely positive trace-preserving (CPTP) maps, i.e., a data processing inequality under the action of quantum channels, that guarantee operational meaning and allow deriving an entire spectrum of information measures. Examples include the Petz–Rényi relative entropy \cite{petz_quasi-entropies_1985, petz_quasi-entropies_1986, tomamichel_fully_2009}, the sandwiched Rényi relative entropy \cite{muller-lennert_quantum_2013, wilde_strong_2014},  the geometric  R\'{e}nyi relative entropy \cite{petz_contraction_1998, matsumoto_new_2018} and (smooth) max/min-relative entropies \cite{renner_security_2006}. Note that the relative entropy also satisfies a data processing inequality and thus many information quantities in the asymptotic setting can be expressed by this specific divergence. Replacing the relative entropy by other divergences in similar expressions allows one to define generalized quantities that are applicable in other regimes, e.g., the one-shot setting.
\\
Despite their broad applicability, evaluating these divergence–based measures poses significant computational problems in the case of high-dimensional quantum states or if their definition involves challenging or infeasible optimizations. In this work, we define several types of divergence-based information measures and prove their invariance under local isometric or unitary transformations.
\\
In Sec.~\ref{sec:localInvs} we introduce the setting and notation and define the types of generalized divergence-based information quantities. We continue to present and prove the main result of this contribution, namely the invariance of information quantities of the defined types under local isometric or unitary transformations. Finally, we conclude our results in Sec.~\ref{sec:conclusion} and provide an outlook to potential applications and future research directions.

%% file: 2_local_invariances.tex
\section{Local invariance of information quantities} \label{sec:localInvs}
In this section we briefly introduce the notation and necessary objects to define several types of generalized information quantities and prove their invariance under local isometric or unitary transformations.
\subsection{Notation and setting}
\label{sec:setting}
We consider two parties $A$ and $B$ with corresponding Hilbert space $\hs = \hs_A \otimes \hs_B$. $L_{+}(\hs)$ and $\denop(\hs)$ denote the spaces of positive-semidefinite operators and density operators (i.e. quantum states) acting on $\hs$, respectively. 
Isometries are denoted by capital $V$. Corresponding channels, i.e., completely positive and trace-preserving maps, are written as $\V(\cdot)  \equiv V~(\cdot)~ V^\dagger$. Unitary operators are named by capital $U$. The identity map is written as $\mathrm{id}$ and $\circ$ denotes the composition of maps.
States or operators are labeled by the systems they act on. For any multipartite state or operator, e.g., the state $\rho_{AB}$, the marginal state or operator is given by $\rho_A := \Tr_B[\rho_{AB}]$. 
\subsection{Generalized divergence-based types of information quantities}
 The information quantities we analyze in this contribution are defined via so-called generalized divergences $\mathbf{D}$ \cite{polyanskiy_arimoto_2010, sharma_fundamental_2013}.
\begin{mydef}[Generalized divergence, data-processing inequality] \label{def:gedDiv} \ \\
    For quantum states $\rho \in \denop(\hs)$ and positive-semidefinite operators $\zeta \in L_+(\hs)$, a function $\mathbf{D}:\denop(\hs) \times L_{+}(\hs) \rightarrow \mathbbm{R} \cup \infty$ that satisfies the data-processing inequality under any channel $\mathcal{N}$, $\mathbf{D}(\rho||\zeta) \geq \mathbf{D}(\mathcal{N}(\rho)||\mathcal{N}(\zeta))$,
    is called (generalized) divergence.
\end{mydef} \noindent
For quantum states, the data-processing inequality can be interpreted as the property that no physical transformation can make two states more distinguishable.
It is known that any divergence is invariant under the application of isometric transformations $\V$ \cite{khatri_principles_2024}:
\begin{flalign}
    \label{eq:gendivIsoInv}
     \mathbf{D}(\V(\rho)||\V(\zeta)) = \mathbf{D}(\rho||\zeta).
\end{flalign}
Generalized divergences are used to define generalized information quantities for bipartite quantum states $\rho_{AB}$. There exist several ways to define information quantities based on divergences with operational meaning. In this work, we consider quantities that relate a bipartite state and its local marginals.
Using specific instances for the divergences, these quantities can be related to various information processing tasks (cf. \cite{khatri_principles_2024} for a comprehensive overview). The method to prove local invariance of the quantities we present in this work, however, does not depend on the specific form of the divergence, but only on how it is applied to the state. We therefore define several types of information quantities independent of the specific form of the divergence.
Some types involve a so-called smoothing, i.e., an optimization over an environment of the quantum state. For this, the sine distance \cite{rastegin_sine_2006} is used as a distance measure that is closely related to the fidelity \cite{uhlmann_transition_1976} as defined below. \\
In the following, let $\varepsilon \in [0,1]$ and let $\rho$ and $\sigma$ be states.
\begin{mydef}[Fidelity, sine distance]\label{def:fidelitySineDistance} \ \\
 The fidelity of two quantum states $\rho$ and $\sigma$ is defined as follows:
 \begin{flalign}
     \mathcal{F}(\rho, \sigma) := \left(\Tr[\sqrt{\sqrt{\sigma} \rho \sqrt{\sigma}}]\right)^2
 \end{flalign}
 The sine distance is defined as:
 \begin{flalign}
     P(\rho,\sigma) := \sqrt{1-\mathcal{F}(\rho,\sigma)}     
 \end{flalign}
\end{mydef} \noindent
However, other distinguishability functions potentially with different properties can be used as well. In this work, we only require that the function obeys the data-processing inequality and invariance under isometric transformations. As both properties are satisfied by any divergence (cf. Definition \ref{def:gedDiv} and \eqref{eq:gendivIsoInv}), including the sine distance~\cite{khatri_principles_2024}, the smoothing environment $B^{\varepsilon}(\rho)$ around a state $\rho$ is defined as follows:
\begin{mydef}[Smoothing environment] \label{def:smoothinBall} \ \\ 
Let $\mathbf{D}$ be any divergence.
\begin{flalign}
    \label{eq:smoothingBall}
    B^\varepsilon(\rho) := \left\{ \hat{\rho} : \mathbf{D}(\rho||\hat{\rho}) \leq \varepsilon \right\} 
\end{flalign}
\end{mydef} \noindent
For the sine distance, it specifically reads:
\begin{flalign}
        \label{eq:smoothBallSine}
        B_P^\varepsilon(\rho) := \left\{ \hat{\rho} : P(\rho,\hat{\rho}) \leq \varepsilon \right\} 
        = \left\{ \hat{\rho} : \mathcal{F}(\rho, \hat{\rho}) \geq 1-\varepsilon^2 \right\}.
    \end{flalign}
 \noindent
In this work, we define and analyze the following types of information quantities:
\begin{mydef}[Types of generalized mutual information] \label{def:typesMutInf}
\begin{flalign}  
    \mathbf{I}_1(\rho_{AB}) &:= \mathbf{D}(\rho_{AB}||\rho_A \otimes \rho_B) \\
    \mathbf{I}_2(\rho_{AB}) &:= \inf_{\sigma_B} \mathbf{D}(\rho_{AB}||\rho_A \otimes \sigma_B) \\
    \mathbf{I}^{\varepsilon}_3(\rho_{AB}) &:= \inf_{\substack{\hat{\rho}_{AB} \in B^{\varepsilon}(\rho_{AB}) \\ \sigma_B}}\mathbf{D}(\hat{\rho}_{AB}||\rho_A \otimes \sigma_B) \\
    \mathbf{I}^{\varepsilon}_4(\rho_{AB}) &:= \inf_{\hat{\rho}_{AB} \in B^{\varepsilon}(\rho_{AB})}\mathbf{D}(\hat{\rho}_{AB}||\rho_A \otimes \hat{\rho}_B)
    \end{flalign}
\end{mydef}
\begin{mydef}[Types of generalized conditional entropies] \label{def:typesCondEnt}
\begin{flalign}
    \mathbf{H}_1(\rho_{AB}) &:= -\mathbf{D}(\rho_{AB}||\id_A \otimes \rho_B) \\
    \mathbf{H}_2(\rho_{AB}) &:= -\inf_{\sigma_B} \mathbf{D}(\rho_{AB}||\id_A \otimes \sigma_B) \\
    \mathbf{H}^{\varepsilon}_3(\rho_{AB}) &:= -\inf_{\substack{\hat{\rho}_{AB} \in B^{\varepsilon}(\rho_{AB}) \\ \sigma_B}} \mathbf{D}(\hat{\rho}_{AB}||\id_A \otimes \sigma_B) 
\end{flalign}    
\end{mydef} \noindent
The defined types include information quantities in both the asymptotic and the one-shot setting. Examples are the (generalized) quantum mutual information ($\mathbf{I}_1, \mathbf{I}_2$), the coherent information ($\mathbf{H}_2$), the Petz-R\'{e}nyi mutual information, the sandwiched R\'{e}nyi mutual information, the geometric R\'{e}nyi mutual information, the smooth min-mutual information (see Ref.~\cite{datta_min-_2009} for these quantities), the $\varepsilon$-hypothesis testing mutual information ($\mathbf{I}^{\varepsilon}_3$), the smooth max-conditional and min-conditional entropy ($\mathbf{H}^{\varepsilon}_3$) and the smooth max-mutual information ($\mathbf{I}^{\varepsilon}_4$). 
We show that the types $\mathbf{I}_1, \mathbf{I}_2$ and $\mathbf{I}^{\varepsilon}_3$ are invariant under any local isometric transformation, while the types $\mathbf{I}^{\varepsilon}_4, \mathbf{H}_1, \mathbf{H}_2$ and $\mathbf{H}^{\varepsilon}_3$ are invariant if the first subsystem is transformed unitarily, while the second subsystem can be transformed by any isometry. More precisely, we prove below:
\begin{prop}
\label{thm:typesLocIsoInv} \ \\
Any information quantity of type $\mathbf{I}_1, \mathbf{I}_2$ and $\mathbf{I}^{\varepsilon}_3$ as in Definition~\ref{def:typesMutInf} is invariant under local isometric transformations of the form 
$\V_{AB}(\cdot) = V_A \otimes V_B ~(\cdot)~V_A^\dagger \otimes V_B^\dagger$.
\end{prop}
\begin{prop}
\label{thm:typesLocUniIsoInv} \ \\
Any information quantity of type $\mathbf{I}^{\varepsilon}_4, \mathbf{H}_1, \mathbf{H}_2$ and $\mathbf{H}^{\varepsilon}_3$ as in Def.~\ref{def:typesMutInf} and Def.~\ref{def:typesCondEnt} is invariant under local unitary transformations in the first system and isometric transformations in the second system of the form 
$\V_{AB}(\cdot) = U_A \otimes V_B ~(\cdot)~U_A^\dagger \otimes V_B^\dagger$.
\end{prop} \noindent
The proofs for Proposition \ref{thm:typesLocIsoInv} and \ref{thm:typesLocUniIsoInv} use the technical Lemmas \ref{thm:revChannelLemma1}, \ref{thm:revChannelLemma2} and \ref{thm:revChannelLemma3} related to the isometric transformations and the corresponding so-called reversal channel \cite{wilde_quantum_2013}, defined in the following. Due to $VV^\dagger \neq \mathbbm{1}$ in general, the map $\V^\dagger$ is not necessarily trace-preserving, although it is completely positive. The reversal channel $\RV$ corresponding to an isometric channel $\V$ is a completely positive and trace-preserving map satisfying $\RV \circ \V = \mathrm{id}$, hence reversing the action of an isometric channel.
\begin{mydef}[Reversal channel]\ \\
Let $\V:L_+(\hs) \rightarrow L_+(\tilde{\hs})$ be an isometric channel and $\omega \in \denop(\hs)$. The reversal channel corresponding to $\V$ and $\omega$ is defined as
\begin{flalign}
    \RV(\sigma):= \V^{\dagger}(\sigma) + \Tr[(\mathbbm{1} -VV^\dagger)\sigma]\omega~,
\end{flalign}
for any $\sigma \in L_+(\tilde{\hs})$.
\end{mydef} \noindent
Note, that this channel is not unique and $\V(\RV(\rho)) \neq \rho$, in general. Also note that while $\sigma \in L_+(\tilde{\hs})$, $\omega \in \denop(\hs)$ needs to be a quantum state such that the map is trace-preserving. Given an isometric channel  $\V$, we define the following sets that relate a smoothing environment to its image and pre-image under the isometric channel and its reversal channel:
\begin{mydef}
\label{def:isorevsets}
\ \\
For a state $\rho$ with $\varepsilon$-environment $B^\varepsilon(\rho)$ as in Definition \ref{def:smoothinBall}, an isometric channel $\V$, and the reversal channel $\RV$, we define:
\begin{flalign}
    B^{\varepsilon}_{\V}(\rho) &:= \lbrace \V(\hat{\rho}) ~|~ \hat{\rho} \in B^{\varepsilon}(\rho) \rbrace, \\
    B^{\varepsilon}_{\RV}(\rho) &:= \lbrace \tilde{\rho} ~|~ \RV(\tilde{\rho}) \in B^\varepsilon (\rho) \rbrace .
\end{flalign}
\end{mydef} 
\noindent We now show that these sets can be ordered (see Figure \ref{fig:smoothing}), implying inequalities for smoothed information quantities as used in the proofs of Propositions \ref{thm:typesLocIsoInv} and \ref{thm:typesLocUniIsoInv}.
\begin{figure}[h]
    \centering
    \vspace{-1em}
    \includegraphics[width=0.8\linewidth]{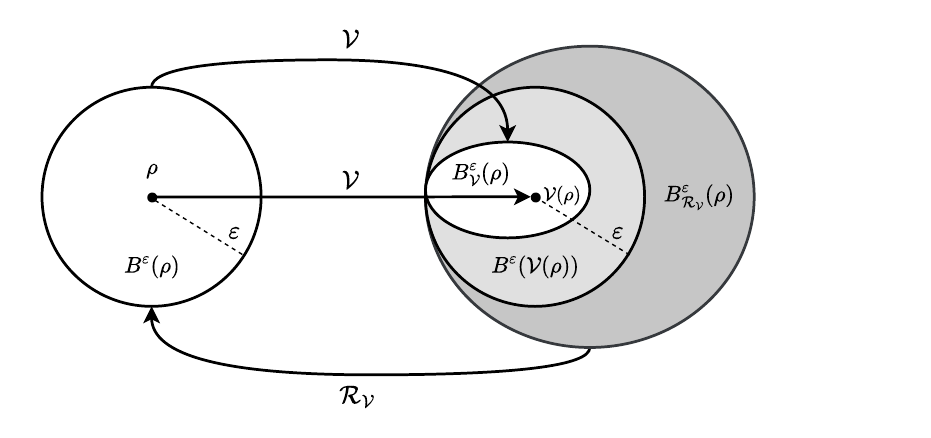}
    \caption{Schematic visualization of the smoothing environment $B^\varepsilon$ and its image $B_{\V}^\varepsilon$ and pre-image $B_{\RV}^\varepsilon$ under the isometric channel $\V$ and its reversal channel $\RV$.}
    \label{fig:smoothing}
\end{figure}
\begin{lemma}
\label{thm:revChannelLemma1}
\ \\
\noindent Regarding the sets from Definitions \ref{def:smoothinBall} and \ref{def:isorevsets}, the following set relations hold:
\begin{flalign}
     B^{\varepsilon}_{\V}(\rho) \subset B^{\varepsilon}(\V(\rho)) \subset B^{\varepsilon}_{\RV}(\rho).
\end{flalign}
\end{lemma}
\begin{proof}
    For the first relation let $\tilde{\rho} \in B^\varepsilon_{\V}(\rho),~\tilde{\rho} = \V(\hat{\rho}),~ \hat{\rho} \in B^\varepsilon(\rho)$. Using the invariance under isometric transformations \eqref{eq:gendivIsoInv}, we have $ \mathbf{D}(\V(\rho)|| \tilde{\rho}) = \mathbf{D}(\V(\rho)|| \V(\hat{\rho})) = \mathbf{D}(\rho|| \hat{\rho}) \leq \varepsilon \implies \tilde{\rho} \in B^\varepsilon(\V(\rho))$.
    Now let $\tilde{\rho} \in B^\varepsilon(\V(\rho)) \implies \mathbf{D}(\V(\rho)||\tilde{\rho}) \leq \varepsilon$. Using the data-processing inequality for any quantum channel, we have $\mathbf{D}(\RV(V(\rho))|| \RV(\tilde{\rho})) \leq \varepsilon \implies \mathbf{D}(\rho|| \RV(\tilde{\rho})) \leq \varepsilon \implies \tilde{\rho} \in B^{\varepsilon}_{\RV}(\rho)$. 
\end{proof}
\noindent A second technical property of the reversal channels is used for proving invariance under local transformations. There exist reversal channels for local transformations that preserve the reversal property $\RV \circ \V = \mathrm{id}$ locally, in the sense $\RV(\V_A(\cdot) \otimes (\cdot)) \sim \mathrm{id}_A \otimes (\cdot)$. \\
For mutual information quantities of types $\mathbf{I}$ using quantum states as arguments of the divergence, the following fact is used for proving Proposition \ref{thm:typesLocIsoInv}.
\begin{lemma} \label{thm:revChannelLemma2} \ \\
    Let $V = V_A \otimes V_B$ be a local isometry. For any state $\rho_A$ and any local quantum states $\sigma_A,~ \sigma_B$, there exists a reversal channel $\RV$ satisfying
    \begin{flalign}
        \RV(\V_A(\rho_A) \otimes \sigma_B) = \rho_A \otimes \Tr_A[\RV(\V_A(\sigma_A) \otimes \sigma_B)].
    \end{flalign}
\end{lemma}
\begin{proof}
    Define the reversal channel with $\omega = \rho_A \otimes \omega_B$ for any local quantum state $\omega_B$. Using the definition of $\RV$ and $V_A^\dagger V_A = \mathbbm{1}$, one can directly calculate:
    \begin{flalign}
        \RV(\V_A(\rho_A) \otimes \sigma_B) = \rho_A \otimes (V_B^\dagger \sigma_B V_B + (1-\Tr_B[V_B V_B^\dagger \sigma_B])\omega_B) = \rho_A \otimes \Tr_A[\RV(\V_A(\sigma_A)\otimes\sigma_B)].
    \end{flalign}
\end{proof}
\noindent For conditional entropies of types $\mathbf{H}$ using both quantum states and positive-semidefinite operators as arguments of the divergence, the following Lemma is used for proving invariance in the proof of Proposition \ref{thm:typesLocUniIsoInv}.
\begin{lemma} \label{thm:revChannelLemma22} \ \\
    Let $V = U_A \otimes V_B$ with unitary $U_A$ and isometric $V_B$. For any local quantum state $\sigma_B$, there exists a reversal channel $\RV$ satisfying
    \begin{flalign}
        \RV(\id_A \otimes \sigma_B) = \id_A \otimes \Tr_A[\RV(\pi_A \otimes \sigma_B)],
    \end{flalign}
    where $\pi_A = \frac{1}{d_A} \id_A$ is the maximally mixed state of $A$ with dimension $d_A$.
\end{lemma}
\begin{proof}
    Define the reversal channel with $\omega = \pi_A \otimes \omega_B$ for any local quantum state $\omega_B$. Considering $U_A^\dagger U_A = U_AU_A^\dagger = \id_A$, using the definition of $\RV$ and calculating as in the proof of Lemma~ \ref{thm:revChannelLemma2}, one finds:
    \begin{flalign}
        \RV(\id_A\otimes \sigma_B) = \id_A \otimes (V_B^\dagger \sigma_B V_B + (1-\Tr_B[V_B V_B^\dagger \sigma_B])\omega_B) = \id_A \otimes \Tr_A[\RV(\pi_A\otimes\sigma_B)].
    \end{flalign}
\end{proof}
\noindent In the special case of smoothed states present in both arguments of the divergence as in $\mathbf{I}^{\varepsilon}_4$, the following property is used for proving  Proposition \ref{thm:typesLocUniIsoInv}.
\begin{lemma} \label{thm:revChannelLemma3} \ \\
    Let $V = U_A \otimes V_B$ with unitary $U_A$ and isometric $V_B$. For any state $\rho_A$ and any quantum state $\sigma_{AB}$, there exists a reversal channel $\RV$ satisfying
    \begin{flalign}
        \RV(\V_A(\rho_A) \otimes \sigma_B) = \rho_A \otimes \Tr_A[\RV(\sigma_{AB})].
    \end{flalign}
\end{lemma}
\begin{proof}
    Set $\omega = \rho_A \otimes \omega_B$ for any quantum state $\omega_B$. Using the definition of $\RV$ and $U_A^\dagger U_A = \mathbbm{1}$, we then have for any local isometry $V = U_A \otimes V_B$:
    \begin{flalign}
        \RV(\V_A(\rho_A)\otimes \sigma_B) &= \rho_A \otimes V_B^\dagger \sigma_B V_B + \Tr[\mathbbm{1}_{AB} - \V_A(\rho_A) \otimes V_B V_B^\dagger \sigma_B]\rho_A \otimes \omega_B \\
        &= \rho_A \otimes (V_B^\dagger \sigma_B V_B + (1 - \Tr_B[V_B V_B^\dagger \sigma_B]\omega_B).
    \end{flalign}
On the other hand, using $U_A U_A^\dagger = \mathbbm{1}_A$, one finds
\begin{flalign}
    \rho_A \otimes \Tr_A[\RV(\sigma_{AB})] &= \rho_A \otimes \Tr_A[
        U_A^\dagger \otimes V_B^\dagger ~\sigma_{AB}~ U_A \otimes V_B 
        + \Tr[(\mathbbm{1}_{AB} - U_A U_A^\dagger \otimes V_B V_B^\dagger) ~\sigma_{AB}
    ] \rho_A \otimes \omega_B] \\
     &= \rho_A \otimes \Tr_A[
        \mathbbm{1}_A \otimes V_B^\dagger ~\sigma_{AB}~ \mathbbm{1}_A \otimes V_B 
        + (1-\Tr[\mathbbm{1}_A \otimes V_B V_B^\dagger ~\sigma_{AB}]) \rho_A \otimes \omega_B        
    ] \\
    &= \rho_A \otimes (V_B^\dagger \sigma_B V_B + (1- \Tr_B[V_BV_B^\dagger~ \sigma_B]) \omega_B),
\end{flalign}
showing the claimed equality.
\end{proof} \noindent
Using these results, we now prove Propositions \ref{thm:typesLocIsoInv} and \ref{thm:typesLocUniIsoInv}.
\begin{proof}[Proof of Proposition \ref{thm:typesLocIsoInv}]
    Let $\V = \V_A \otimes \V_B$, where $\V_{(A/B)}: \denop(\hs_{(A/B)}) \rightarrow \denop(\tilde{\hs}_{(A/B)})$. Let $im(\V_{(A/B)}) \subset \denop(\tilde{\hs}_{(A/B)})$ denote the image of $\V_{(A/B)}$. \\
    First consider type $\mathbf{I}_1$. Using $\Tr_{(A/B)}[\V_A \otimes \V_B(\rho_{AB})] = \V_{(B/A)}(\rho_{B/A})$ and the invariance of any divergence under isometric evolution \eqref{eq:gendivIsoInv}, one has:
    \begin{flalign}
        \mathbf{I}_1(\V(\rho_{AB})) &= \mathbf{D}(\V(\rho_{AB}) || \Tr_B[\V(\rho_{AB})] \otimes \Tr_A[V(\rho_{AB})]) \\
        &= \mathbf{D}(\V(\rho_{AB}) || \V(\rho_A \otimes \rho_B)) \\
        &= \mathbf{D}(\rho_{AB} || \rho_A \otimes \rho_B) \\
        &= \mathbf{I}_1(\rho_{AB}).
    \end{flalign}
    Next, consider type $\mathbf{I}_2$. Again using the invariance under isometric transformations, one has:
    \begin{flalign}
        \label{eq:ineq51}
        \mathbf{I}_2(\rho_{AB}) &= \inf_{\sigma_B \in \denop ( \hs_B)} \mathbf{D} (\rho_{AB} || \rho_A \otimes \sigma_B) \\
        &= \inf_{\sigma_B \in \denop ( \hs_B)} \mathbf{D} (\V(\rho_{AB}) ||\V_A(\rho_A) \otimes \V_B(\sigma_B)) \\
        &= \inf_{\tilde{\sigma}_B \in im(\V_B)} \mathbf{D} (\V(\rho_{AB}) ||\V_A(\rho_A) \otimes \tilde{\sigma}_B) \\
        &\geq \inf_{\tilde{\sigma}_B \in \denop(\tilde{\hs}_B)} \mathbf{D} (\V(\rho_{AB}) ||\V_A(\rho_A) \otimes \tilde{\sigma}_B) \\
        &= \mathbf{I}_2(\V(\rho_{AB}))
    \end{flalign}
    Conversely, using the data processing inequality for $\mathbf{D}$ with the reversal channel $\RV$, one has:
    \begin{flalign}
        \mathbf{I}_2(\V(\rho_{AB})) &= \inf_{\tilde{\sigma}_B \in \denop(\tilde{\hs}_B)} \mathbf{D}(\V(\rho_{AB}) || \Tr_B[\V(\rho_{AB})] \otimes \tilde{\sigma}_B) \\
        &= \inf_{\tilde{\sigma}_B \in \denop(\tilde{\hs}_B)} \mathbf{D}(\V(\rho_{AB}) || \V_A(\rho_A) \otimes \tilde{\sigma}_B) \\
        &\geq \inf_{\tilde{\sigma}_B \in \denop(\tilde{\hs}_B)} \mathbf{D}(\rho_{AB} || \RV(\V_A(\rho_A) \otimes \tilde{\sigma}_B))
    \end{flalign}
        Using the reversal channel as in Lemma \ref{thm:revChannelLemma2}, with arbitrary $\hat{\sigma}_A \in \denop ( \hs_A)$ yields therefore:
    \begin{flalign}
        \label{eq:ineq52}
        \mathbf{I}_2(\V(\rho_{AB})) &\geq \inf_{\tilde{\sigma}_B \in \denop(\tilde{\hs}_B)} \mathbf{D}(\rho_{AB} || \rho_A \otimes \Tr_A [\RV(\V_A(\hat{\sigma}_A) \otimes \tilde{\sigma}_B)])\\
        &\geq \inf_{\sigma_B \in \denop ( \hs_B)} \mathbf{D}(\rho_{AB} || \rho_A \otimes \sigma_B) \\
        &= \mathbf{I}_2(\rho_{AB}).
    \end{flalign}
    Note, that the second inequality holds because $\Tr_A [\RV(\V_A(\hat{\sigma}_A) \otimes \tilde{\sigma}_B)] \in \denop(\hs_B)$.
    The inequalities below \eqref{eq:ineq51} and  \eqref{eq:ineq52} together imply the claimed equality.\\
    Finally, consider type $\mathbf{I}^{\varepsilon}_3$. Using similar arguments and Lemma \ref{thm:revChannelLemma1} one finds:
    \begin{flalign}
        \mathbf{I}^{\varepsilon}_3(\rho_{AB}) &= \inf_{\substack{\hat{\rho}_{AB} \in B^{\varepsilon}(\rho_{AB}) \\\sigma_B \in \denop(\hs_B)}}
        \mathbf{D}(\hat{\rho}_{AB}||\rho_A \otimes \sigma_B)  \\ 
        &= \inf_{\substack{\hat{\rho}_{AB} \in B^{\varepsilon}_{\V}(\rho_{AB}) \\ \tilde{\sigma}_B \in im(\V_B)}}
        \mathbf{D}(\hat{\rho}_{AB}||\V_A(\rho_A) \otimes \tilde{\sigma}_B) \\
        &\geq \inf_{\substack{\tilde{\rho}_{AB} \in B^{\varepsilon}(\V(\rho_{AB})) \\ \tilde{\sigma}_B \in \denop(\tilde{\hs}_B)}}
        \mathbf{D}(\tilde{\rho}_{AB}||\V_A(\rho_A) \otimes \tilde{\sigma}_B) \\
        &= \mathbf{I}^{\varepsilon}_3(\V(\rho_{AB}))
    \end{flalign}
    Conversely, using again the reversal channel as in Lemma \ref{thm:revChannelLemma2} with arbitrary $\hat{\sigma}_A$ together with Lemma \ref{thm:revChannelLemma1} and noting the set equality $\RV[B^\varepsilon_{\RV}(\rho_{AB})] = B^\varepsilon(\rho_{AB})$, one concludes:
    \begin{flalign}
        \mathbf{I}^{\varepsilon}_3(\V(\rho_{AB})) &= \inf_{\substack{\tilde{\rho}_{AB} \in B^{\varepsilon}(\V(\rho_{AB})) \\ \tilde{\sigma}_B \in \denop(\tilde{\hs}_B)}}
        \mathbf{D}(\tilde{\rho}_{AB}||\V_A(\rho_A) \otimes \tilde{\sigma}_B) \\ 
        &\geq \inf_{\substack{\tilde{\rho}_{AB} \in B^{\varepsilon}_{\RV}(\V(\rho_{AB})) \\ \tilde{\sigma}_B \in \denop(\tilde{\hs}_B)}}
        \mathbf{D}(\tilde{\rho}_{AB}||\V_A(\rho_A) \otimes \tilde{\sigma}_B) \\
        &\geq \inf_{\substack{\tilde{\rho}_{AB} \in B^{\varepsilon}_{\RV}(\V(\rho_{AB})) \\ \tilde{\sigma}_B \in \denop(\tilde{\hs}_B)}}
        \mathbf{D}(\RV(\tilde{\rho}_{AB})||\RV(\V_A(\rho_A) \otimes \tilde{\sigma}_B)) \\
        &= \inf_{\substack{\hat{\rho}_{AB} \in B^{\varepsilon}(\rho_{AB}) \\ \tilde{\sigma}_B \in \denop(\tilde{\hs}_B)}}
        \mathbf{D}(\hat{\rho}_{AB}||\rho_A \otimes \Tr_A [\RV(\V_A(\hat{\sigma}_A) \otimes \tilde{\sigma}_B)]) \\
        &\geq \inf_{\substack{\hat{\rho}_{AB} \in B^{\varepsilon}(\rho_{AB}) \\ \sigma_B \in \denop(\hs_B)}}
        \mathbf{D}(\hat{\rho}_{AB}||\rho_A \otimes \sigma_B) \\
        &= \mathbf{I}^{\varepsilon}_3(\rho_{AB})
    \end{flalign}
\end{proof}

\begin{proof}[Proof of Proposition \ref{thm:typesLocUniIsoInv}] 
    Let $\V = \V_A \otimes \V_B,~\V_A(\cdot) = U_A~(\cdot)~U_A^\dagger$, $\V_A:  \denop(\hs_A) \mapsto \denop(\hs_A)$ with unitary $U_A$ and $\V_B: \denop(\hs_B) \mapsto \denop(\tilde{\hs}_B)$ with isometric $\V_B$.
    Let $im(\V_B) \subset \denop(\tilde{\hs}_B)$ denote the image of $\V_B$. \\
    First, consider the type $\mathbf{I}^{\varepsilon}_4$. Using the invariance of $\mathbf{D}$ under isometric transformations, and Lemma \ref{thm:revChannelLemma1}, one has:
    \begin{flalign}
        \mathbf{I}^{\varepsilon}_4(\rho_{AB}) &= \inf_{\hat{\rho}_{AB} \in B^\varepsilon(\rho_{AB})} \mathbf{D} (\hat{\rho}_{AB} || \rho_A \otimes \hat{\rho}_B) \\
        &= \inf_{\hat{\rho}_{AB} \in B^\varepsilon(\rho_{AB})} \mathbf{D} (\V(\hat{\rho}_{AB}) || \V(\rho_A \otimes \hat{\rho}_B)) \\
        &= \inf_{\tilde{\rho}_{AB} \in B_\V^\varepsilon(\rho_{AB})} \mathbf{D} (\tilde{\rho}_{AB} || \V_A(\rho_A) \otimes \tilde{\rho}_B) \\
        &\geq \inf_{\tilde{\rho}_{AB} \in B^\varepsilon(\V(\rho_{AB}))} \mathbf{D} (\tilde{\rho}_{AB} || \V_A(\rho_A) \otimes \tilde{\rho}_B) \\
        &= \mathbf{I}^{\varepsilon}_4(\V(\rho_{AB})).
    \end{flalign}
    Conversely, using again Lemma \ref{thm:revChannelLemma1} and applying the data processing inequality for the reversal channel $\RV$ as in Lemma \ref{thm:revChannelLemma3}, one finds:
    \begin{flalign}
        \mathbf{I}^{\varepsilon}_4(\V(\rho_{AB})) &= \inf_{\tilde{\rho}_{AB} \in B^\varepsilon(\V(\rho_{AB}))} \mathbf{D} (\tilde{\rho}_{AB} || \V_A(\rho_A) \otimes \tilde{\rho}_B) \\
        &\geq \inf_{\tilde{\rho}_{AB} \in B_{\RV}^\varepsilon(\rho_{AB})} \mathbf{D} (\tilde{\rho}_{AB} || \V_A(\rho_A) \otimes \tilde{\rho}_B) \\
        &\geq  \inf_{\tilde{\rho}_{AB} \in B_{\RV}^\varepsilon(\rho_{AB})} \mathbf{D} (\RV(\tilde{\rho}_{AB}) || \RV(\V_A(\rho_A) \otimes \tilde{\rho}_B)).
    \end{flalign}
    Now since $U_A$ is unitary, Lemma \ref{thm:revChannelLemma3} implies together with the set equality $\RV[B^\varepsilon_{\RV}(\rho_{AB})] = B^\varepsilon(\rho_{AB})$,
    \begin{flalign}
        \mathbf{I}^{\varepsilon}_4(\V(\rho_{AB}))  &\geq \inf_{\tilde{\rho}_{AB} \in B_{\RV}^\varepsilon(\rho_{AB})} \mathbf{D} (\RV(\tilde{\rho}_{AB}) || \rho_A \otimes \Tr_A[\RV(\tilde{\rho}_{AB})]) \\
        &= \inf_{\hat{\rho}_{AB} \in B^\varepsilon(\rho_{AB})} \mathbf{D}(\hat{\rho}_{AB} || \rho_A \otimes \hat{\rho}_B) \\
        &= \mathbf{I}^{\varepsilon}_4(\rho_{AB}),
    \end{flalign}
    proving the claimed property.\\
    Finally, consider the conditional entropy types $\mathbf{H}_1, \mathbf{H}_2$ and $\mathbf{H}^{\varepsilon}_3$. Noting $\V_A(\id_A) = \id_A$ due to the fact that $\V_A$ is unitary, the proofs of local invariance are equivalent to those for $\mathbf{I}_1, \mathbf{I}_2$ and $\mathbf{I}^{\varepsilon}_3$ given in the proof of Proposition \ref{thm:typesLocIsoInv}, respectively, if $\rho_A$ is replaced by $\id_A$ and Lemma \ref{thm:revChannelLemma22} is used instead of Lemma \ref{thm:revChannelLemma2}.
\end{proof}

%% file: 3_conclusion.tex
\section{Conclusion}
\label{sec:conclusion}
In this work we defined types of quantum information quantities based on generalized divergences and analyzed their properties regarding local isometric and unitary transformations. Leveraging properties of the reversal channel associated with the local transformation, we proved that several types of information quantities imply invariance under any local isometric transformation, while others are invariant if one of the subsystems is transformed by a unitary. 
Two main technical results are utilized to prove the local invariance. First, the smoothing environment required by some types of information quantities is related to its image and its pre-image under the isometric channel and a corresponding reversal channel, implying the set relation of Lemma \ref{thm:revChannelLemma1}. Second, the action of the reversal channel for a local isometric channel is characterized in Lemmas \ref{thm:revChannelLemma2}-\ref{thm:revChannelLemma3}. Combining these insights with the data processing inequality of any generalized divergence allows us to derive the main results regarding invariance for several types of information quantities under local isometric or unitary transformation as presented in Proposition \ref{thm:typesLocIsoInv} and \ref{thm:typesLocUniIsoInv}.
\\
Since methods for proving the invariance do not depend on the specific form of the generalized divergence used to define the information quantity, the invariances hold for numerous key quantities with operational relevance in quantum information processing.
These invariances can enable more efficient computation of quantum information quantities by reducing complex states to simpler, equivalent forms with symmetries or smaller dimensions. For instance, consider the setting of quantum cryptography including a third party holding the purifying system of a bipartite state. If the output of a tripartite quantum channel is a pure state, it is equivalent to the purified output of a corresponding bipartite channel up to a local isometry in the purifying system. The dimension of the purification system of the latter state, however, may be smaller than the dimension of the former purification system. Consequently, the evaluation of relevant information quantities like the private information may be significantly improved by calculating it for the smaller but equivalent output state. Such invariances reduce computational overhead and unlock flexibility in protocol design. Given that the protocol performance is measured by an invariant information quantity, local unitary or isometric operations like encoding or correction operations can be added, modified or removed while preserving performance metrics.
\\
Note, that the general proofs for local invariance depend only on the monotonicity by the data-processing inequality and the reversal properties of the isometric channel as given in Lemma \ref{thm:revChannelLemma2}--\ref{thm:revChannelLemma3}. Therefore, local invariance may be shown similarly for other channels and corresponding reversals, such as the Petz recovery map (cf.~\cite{khatri_principles_2024}).
\\
The computation of generalized information quantities remains a significant challenge, in general. While some quantities like the hypothesis testing mutual information can be calculated efficiently, e.g., by semidefinite programs \cite{dupuis_generalized_2012,datta_second-order_2016}, others can only be approximated (see Ref.~\cite{nuradha_fidelity-based_2024} regarding the smooth max-relative entropy) or lack a general and efficient solution like the smooth max-mutual information. Our results may enable new methods to bound or approximate these quantities by leveraging state equivalences. 
Finally, we note that the invariance under local isometries, notably local Clifford and unitary operations, plays a pivotal role in the classification and optimization of graph states as used in error-correction \cite{beny_general_2010, sarkar_graph-based_2024} or measurement-based quantum computation \cite{hein_multiparty_2004, briegel_measurement-based_2009}, allowing for the optimization of codes or resource-states with respect to an invariant information measure.
\\
In conclusion, the results of this contribution allow us to derive the behavior of general types of information quantities under local isometric or unitary transformations. Therefore, they improve the capability to characterize and compute information quantities relevant throughout the vast field of quantum information processing.